\documentclass[11pt,a4paper]{article}
\usepackage{indentfirst}
\usepackage[latin1]{inputenc}
\usepackage{mathtools}
\usepackage{amsfonts}
\usepackage[english]{babel}
\usepackage{amssymb}
\usepackage{graphicx}
\usepackage{geometry}
\usepackage{amsthm}
\usepackage{amsmath,bm}
\usepackage{breqn}
\newtheorem{theorem}{Theorem}[section]

\usepackage{titlesec}
\titlelabel{\thetitle.\quad}
\usepackage{titlesec}
\usepackage{hyperref}

\AtBeginDocument{%
}

\begin{document}
	\title{Modified general relativity}
	\author{Gary Nash
		\\University of Alberta, Edmonton, Alberta, Canada,T6G 2R3\\gnash@ualberta.net \footnote{present address Edmonton, AB}
	}
\date{~August 10, 2025}
\maketitle
\vspace{2mm}
\begin{abstract}
Using Einstein's original postulate of a total energy-momentum tensor that includes the energy-momentum of both ordinary matter and the gravitational field, and the orthogonal decomposition of symmetric (0,2) tensors on a paracompact boundaryless time-oriented Lorentzian manifold, a connection-independent symmetric tensor $\varPhi_{\alpha\beta}$ is constructed. It represents the local gravitational energy-momentum that does not exist in General Relativity, and completes the Einstein equation. $\varPhi_{\alpha\beta}$ is constructed from the Lie derivative of both a Lorentzian metric and the regular (nowhere vanishing) unit line element covectors, which intrinsically belong to that metric. Solutions to the Einstein equation in Modified General Relativity (MGR) explicitly depend on the line element covectors in $\varPhi_{\alpha\beta}$, which gives MGR the extra freedom necessary to describe dark energy and dark matter geometrically. $\Phi $, the trace of $ \varPhi_{\alpha\beta} $, dynamically replaces the cosmological constant.

\end{abstract}
\vspace*{0.4cm}
\newline This is a revision of an article published in General Relativity and Gravitation. The original authenticated version is available online at: https://doi.org/10.1007/s10714-019-2537-y.

\section{Introduction}
The local gravitational fields in General Relativity (GR) are described by the Christoffel symbols $\Gamma^{\mu}_{\alpha\beta}$ of the symmetric Levi-Civita connection. During free fall, the Christoffel symbols locally vanish from the existence of Riemann or Fermi normal coordinates, whereby a coordinate system can be found where all of the first derivatives of the metric locally vanish at a point, or along a geodesic, respectively.  The vanishing of the Christoffel symbols means there is no local gravitational field, and therefore, no local gravitational energy-momentum; the equivalence principle forbids it \cite{MTW}. Local gravitational energy-momentum does not exist in GR.
\par Einstein knew that gravity gravitates and initially postulated \cite{EinGross},\cite{Ein16} that the \emph{total} energy-momentum tensor consisted of the energy-momentum of both ordinary matter and the gravitational field. However, his efforts to construct a (connection-independent) symmetric tensor that represents local gravitational energy-momentum were unsuccessful. Einstein's General Relativity was incomplete and remains so today without that tensor. Moreover, GR has no mathematical freedom without that tensor to \emph{fundamentally} describe dark energy and dark matter geometrically.
\par It has been over a century since Einstein developed General Relativity. This article presents the discovery of a symmetric, connection-independent tensor, $\varPhi_{\alpha\beta}$, that describes local gravitational energy-momentum. It is constructed with Lie derivatives and is therefore independent of the Levi-Civita connection $\nabla$. When the connection coefficients (Christoffel symbols) vanish under free fall, $ \varPhi_{\alpha\beta} $ does not vanish because it is equivalently expressed with partial derivatives. That property is \emph{essential} to the existence of local gravitational energy-momentum relative to the equivalence principle. Local gravitational energy-momentum does not vanish during free fall and mysteriously reappear after the event; it is invariant under free fall, defined by $\varPhi_{\alpha\beta}$, and there is no conflict with the equivalence principle.
\par $\varPhi_{\alpha\beta}$ is constructed from the line element vector field, which is now discussed.  Every paracompact differentiable manifold $ \mathcal{M} $ admits a continuous non-zero vector field and thus a continuous line element field \cite{Markus}. That means there exists a differentiable regular (everywhere non-vanishing) vector field $\bm{X}$ collinear to a differentiable non-vanishing unit vector $ \bm{u} $ by $\bm{X}=f\bm{u}$ where $ f\neq0$ is the scalar magnitude of $ \bm{X}$. The line element (or direction) field $(\bm{X},-\bm{X})$ is defined as an assignment of a pair of equal and opposite vectors at each point of $\mathcal{M}$. It is a one-dimensional vector subspace of the tangent space on $ \mathcal{M} $, or equivalently, a line subbundle $l$ of the tangent bundle on $\mathcal{M}$. \par This leads to a technical issue that needs to be addressed. While $\mathcal{M}$ admits a non-vanishing vector field, it does so as a vector field with isolated zeros \cite{Markus}, which can be swept out to infinity. However, the resulting non-vanishing vector field is unsatisfactory because it may not be bounded \cite{kato}. Here, boundedness is guaranteed from the structure of the covectors and vectors (\ref{u}) in $l$ by demanding  $\partial_{\alpha}f$ and $\partial^{\alpha}f$ to be bounded, and $f$ to be bounded. $\Phi$ is bounded because it satisfies (\ref{intPhi}), so the line subbundle $l$ is bounded and is constructed from regular vector fields. This ensures the boundedness of solutions to the Einstein equation in MGR that explicitly depend on the line element vectors in $l$.\par  The line element vectors at each point of a Lorentzian spacetime can be divided into three classes depending on whether $X_{\beta}X^{\beta}$ is negative, positive, or zero: timelike, spacelike, and null, respectively. The null vectors form the familiar double cones in the tangent space on $\mathcal{M}$, which separates the timelike vectors from the spacelike vectors. A Lorentzian metric has a directional characteristic described by the line element vectors associated with the metric. 
\par A Lorentzian metric $g_{\alpha\beta}$ with a +2 signature is related to a Riemannian metric by
\begin{equation}\label{gab}
	g_{\alpha\beta}=g_{\alpha\beta}^{+}-2u_{\alpha}u_{\beta}
\end{equation}where the Riemannian metric $g^{+}_{\alpha\beta}$ is independent of the covectors $u_{\alpha}$ and the unit vectors satisfy $g_{\alpha\beta}u^{\alpha}u^{\beta}=-1$, and $g^{+}_{\alpha\beta}u^{\alpha}u^{\beta}=1$. A paracompact manifold admits a Lorentzian metric $ g_{\alpha\beta} $ if and only if it admits a line element field \cite{CB, Hawk}. Thus, a Lorentzian metric does not exist without the line element field. The regular vectors in the line element field are not introduced arbitrarily as absolute vectors in a Lorentzian spacetime $(\mathcal{M},g_{\alpha\beta})$ because it does not exist without them. 
\par MGR selects a particular $l$ consisting of those unit covectors satisfying (\ref{u}) constrained by (\ref{nu}) from the myriad of covectors that could belong to that $l$. MGR depends explicitly on the Riemannian metric on $\mathcal{M}$ and the collinear vectors \textbf{u} and $\textbf{X}$ in $l$ from the introduction of $\varPhi_{\alpha\beta}$ into the complete Einstein equation. The explicit dependence on the line subbundle in MGR gives it the extra freedom required to describe dark energy, and dark matter without introducing a dark matter profile; higher than second-order derivatives or higher than four-dimensional spacetimes are not required. 
\par Whereas GR is constructed in the tangent bundle on $\mathcal{M}$ that contains a regular timelike vector $\textbf{X}$ and a Riemannian metric. GR describes the gravitation of ordinary matter with its definitive dependence on the Riemannian metric; it does not depend explicitly on $\textbf{X}$.  
\par This article is structured as follows: In section 2, the Orthogonal Decomposition Theorem is proved. In section 3, a modified equation of GR is geometrically derived from Einstein's original postulate, the ODT and Lovelock's theorem, and is dynamically obtained from an action functional.   Important properties of $\varPhi_{\alpha\beta}$ are discussed and variations of the action functional with respect to $g^{\alpha\beta}$, $ u^{\nu} $, and $f$ are presented in Appendix A. Section 4 produces the local conservation equation for the total energy-momentum tensor $T_{\alpha\beta}$ from the diffeomorphic invariance of the action functional. Section 5 is a discussion of the modified Einstein equation of GR in the Friedmann-Lema{\^i}tre-Robertson-Walker (FLRW) metric. Dark energy is defined as the local gravitational energy $\varPhi_{00}$ and dark matter from the stress components $\varPhi_{ii}$. A cyclic universe is discussed.
In section 6, the modified equation of GR is calculated with a spherically symmetric static metric in a region of spacetime free of ordinary matter. Two additional terms appear in the modified Newtonian force equation that provides it the flexibility to describe various types of galaxies. By balancing the dark energy force with the Newtonian force, the Tully-Fisher relation is established. The relation between the acceleration parameter in MOND and the dark energy parameter $a_{0}$ is presented. 
\section{Orthogonal Decomposition of Symmetric Tensors}
The orthogonal decomposition of symmetric tensors on a Riemannian manifold is founded on the classic theorem \cite{Berger} about the orthogonal decomposition of symmetric tensors in the symmetric cotangent bundle $S^{2}T^{*}\mathcal{M}$ on a compact Riemannian manifold. With a different technique, that decomposition was achieved on a paracompact Riemannian manifold without boundary and with smooth sections of $S^{2}T^{*}\mathcal{M}$ with compact support \cite{Gil}. A decomposition of symmetric tensors on a paracompact, time-oriented, boundaryless Lorentzian manifold is presented as the Orthogonal Decomposition Theorem (ODT): 
\begin{theorem}
An arbitrary non-divergenceless (0,2) symmetric tensor $ w_{\alpha\beta} $ in the symmetric cotangent bundle $ S^{2}T^{\ast}\mathcal{M} $ with smooth sections of compact support on an n-dimensional paracompact boundaryless time-oriented Lorentzian manifold $ \mathcal{M} $ with a Levi-Civita connection can be orthogonally decomposed as $
w_{\alpha\beta}= v_{\alpha\beta}+ \varPhi_{\alpha\beta} $ where $v_{\alpha\beta}  $ represents a linear sum of symmetric divergenceless (0,2) tensors and $\varPhi_{\alpha\beta}=\frac{1}{2}\pounds_{X}g_{\alpha\beta}+\pounds_{X}(u_{\alpha}u_{\beta})$ where the timelike unit vector $\bm{u}  $ is collinear with one of the pair of regular vectors in the line element field $ (\bm{X},-\bm{X}) $ and $\bm X$ is not a Killing vector.
\end{theorem}
\begin{proof}
Let the Lorentzian manifold $ (\mathcal{M},g_{\alpha\beta}) $ be paracompact, Hausdorff, time-oriented, and boundaryless. A smooth regular line element field $(\bm{X},\bm{-X)}$ exists \cite{Markus} as does a timelike unit vector $ \bm{u} $ collinear with one of the pair of line element vectors. Let $ \mathcal{M} $ be endowed with a smooth Riemannian metric $ g^{+}_{\alpha\beta} $. The smooth Lorentzian metric $ g_{\alpha\beta} $ is constructed \cite{CB,Hawk} from $ g^{+}_{\alpha\beta} $ and the unit covectors $ u_{\alpha}$ and $ u_{\beta}$ : $g_{\alpha\beta}=g^{+}_{\alpha\beta}-2u_{\alpha}u_{\beta} $. Let $ w_{\alpha\beta} $ and $ v_{\alpha\beta} $ belong to $ S^{2}T^{\ast}\mathcal{M} $, the cotangent bundle of symmetric $(0,2)$ tensors on $ \mathcal{M} $ with smooth sections of compact support. An arbitrary non-divergenceless $ (0,2) $ symmetric tensor $ w_{\alpha\beta} $ in  $ S^{2}T^{\ast}\mathcal{M} $ can be orthogonally and uniquely decomposed \cite{Berger, Gil} according to $ w_{\alpha\beta}=v_{\alpha\beta}+\frac{1}{2}\pounds_{\xi}g^{+}_{\alpha\beta}
$ where $ \bm{\xi} $ is a vector in $T\mathcal{M}$ that is defined up to a Killing vector field and $v_{\alpha\beta}  $ represents a linear sum of symmetric divergenceless (0,2) tensors: $ {\nabla^{+}}^{\alpha}v_{\alpha\beta}=0 $.\par The divergence of $ v_{\beta}^{\alpha} $ in the mixed tensor bundle can be written as $ \nabla_{\alpha}v_{\beta}^{\alpha}=\partial_{\alpha}v_{\beta}^{\alpha}+\frac{v_{\beta}^{\lambda}}{2g}\partial_{\lambda}g-\frac{1}{2}v^{\alpha\lambda}\partial_{\beta}g_{\alpha\lambda} $. Since the determinant of $ g_{\alpha\beta} $, $ g $, is related to that of $ g^{+}_{\alpha\beta} $ by $ g=-g^{+} $ 
\begin{equation}\label{Dv}
	\begin{split}
		\nabla_{\alpha}v_{\beta}^{\alpha}-\nabla^{+}_{\alpha}v_{\beta}^{\alpha}=v^{\alpha\lambda}\partial_{\beta}(u_{\alpha}u_{\lambda}).
	\end{split}
\end{equation} The left hand side of (\ref{Dv}) is a (0,1) tensor but the right hand side is not, which demands
\begin{equation}\label{vuu}
	v^{\alpha\lambda}\partial_{\beta}(u_{\alpha}u_{\lambda})=0.
\end{equation}$v^{\alpha\lambda} $ includes the Lovelock tensors \cite{Love}, all of which reduce to a constant times the Minkowski metric $ \eta^{\alpha\lambda} $ in an n-dimensional flat Lorentzian spacetime. Hence, $ \partial_{\beta}(u_{\alpha}u^{\alpha})=0 $ and (\ref{vuu}) holds for  $\partial_{\beta}u_{\alpha}\neq0 $. Condition (\ref{vuu}) must be satisfied by any non-Lovelock tensor contained in $v^{\alpha\lambda}  $. Thus, (\ref{vuu}) guarantees $\nabla^{\alpha}v_{\alpha\beta}=0  $ for all Lovelock tensors and certain non-Lovelock tensors because $\nabla^{+\alpha}v_{\alpha\beta}=0  $ and it follows that
\begin{equation}\label{decomp}
	w_{\alpha\beta}=v_{\alpha\beta}+\frac{1}{2}\pounds_{\xi}g_{\alpha\beta}+\pounds_{\xi}u_{\alpha}u_{\beta}
\end{equation} where $ \nabla^{\alpha}v_{\alpha\beta}=0 $. $ \xi^{\lambda} $ is chosen to be represented by $ X^{\lambda} $, which is collinear with $ u^{\lambda} $. Using $ X^{\lambda}=fu^{\lambda} $ where $f\neq0  $ is the magnitude of $X^{\lambda}  $, the expression  $X^{\lambda}\nabla_{\lambda}(u_{\alpha}u_{\beta}) $ in the last term of (\ref{decomp}) then vanishes in an affine parameterization and $
w_{\alpha\beta}=v_{\alpha\beta}+\varPhi_{\alpha\beta}	$ where $			\varPhi_{\alpha\beta}:=\frac{1}{2}(\nabla_{\alpha}X_{\beta}+\nabla_{\beta}X_{\alpha})+u^{\lambda}(u_{\alpha}\nabla_{\beta}X_{\lambda}+u_{\beta}\nabla_{\alpha}X_{\lambda})$.
Provided $ w_{\alpha\beta}\neq0 $, the decomposition is orthogonal: $<v_{\alpha\beta},\varPhi_{\alpha\beta}>=0  $. Since $\bm X$ is not a Killing vector, $\varPhi_{\alpha\beta}$ does not vanish.
\end{proof}
\section{Derivation of the modified equation of general relativity}
The geometrical derivation of $\varPhi_{\alpha\beta}$ follows from Einstein's original postulate, the ODT, and Lovelock's Theorem (in a four-dimensional spacetime, the only tensors which are symmetric, divergenceless, and a concomitant of the metric together with its first two derivatives, are the metric and the Einstein tensor).
\par Einstein postulated that the total energy-momentum tensor $T_{\alpha\beta}$ consisted of the energy-momentum from ponderable (ordinary) matter plus that of the gravitational field. This demands the local conservation of total energy-momentum with the vanishing of $ \nabla^{\alpha}T_{\alpha\beta} $ and not $ \nabla^{\alpha}\tilde{T}_{\alpha\beta} $, where $ \tilde{T}_{\alpha\beta} $ is the ordinary matter component of $ T_{\alpha\beta} $. Since $ \tilde{T}_{\alpha\beta} $ is not divergenceless, it can be set with the constant $k$ to an arbitrary non-divergenceless symmetric tensor $ w_{\alpha\beta} $, which is then orthogonally decomposed by the ODT: $w_{\alpha\beta}=k\tilde{T}_{\alpha\beta}=v_{\alpha\beta}+\varPhi_{\alpha\beta} $. $v_{\alpha\beta}$ contains the metric and the
Einstein tensor as a consequence of Lovelock's theorem. Since $T_{\alpha\beta} $ is uniquely determined up to a divergenceless tensor, the sum of possible divergenceless non-Lovelock tensors in $v_{\alpha\beta}$ can be dismissed. Hence, $v_{\alpha\beta}=\Lambda_{C} g_{\alpha\beta}+G_{\alpha\beta}$ where $\Lambda_{C}$ is the cosmological constant. The ODT and Lovelock's theorem generate the modified Einstein equation \emph{in one line} 
\begin{equation}\label{MEQ}
	\Lambda_{C} g_{\alpha\beta}+G_{\alpha\beta}+\varPhi_{\alpha\beta}=k\tilde{T}_{\alpha\beta}
\end{equation}
with $k:=\frac{8\pi G}{c^{4}}$. This is the geometrical basis of Modified General Relativity. MGR exploits a particular Lorentzian metric to construct $ \varPhi_{\alpha\beta} $ from the
ODT: 
\begin{equation}\label{Phiab}
	\begin{split}
		\varPhi_{\alpha\beta}:=	\frac{1}{2}\pounds_{X}g_{\alpha\beta}+\pounds_{X}(u_{\alpha}u_{\beta}).
	\end{split}
\end{equation} The covectors $u_{\alpha}$ and $u_{\beta}$ in (\ref{gab}) are precisely those in (\ref{Phiab}). The flow vector $\textbf{X}$ in the Lie derivatives of (\ref{Phiab}) is collinear with \textbf{u}. Thus, the vectors \textbf{u} and \textbf{X} employed in MGR are directly related to a given Lorentzian metric and are not introduced arbitrarily. $\varPhi_{\alpha\beta}  $ vanishes if and only if $ X^{\mu} $ is a Killing vector. However, in general, there are no Killing vector fields unless a particular symmetry is involved.  
\par Equation (\ref{MEQ}) can be obtained dynamically from the action functional $ S=S^{F}+S^{EH}+S^{G} $, which consists of the action for all ordinary matter fields $S^{F}$, the Einstein-Hilbert action of GR, $S^{EH}$, and the action for the energy-momentum of the gravitational field, $S^{G}$:
\begin{equation}
	\begin{split}
		S=\int L^{F}( A^{\beta},\nabla^{\alpha} A^{\beta},...,g^{\alpha\beta})\sqrt{-g}d^{4}x
		+\frac{c^{3}}{16\pi G} \int (R-2\Lambda)\sqrt{-g}d^{4}x\\-\frac{c^{3}}{16\pi G}\int \varPhi_{\alpha\beta} g^{\alpha\beta}\sqrt{-g} d^{4}x
	\end{split}
\end{equation}
where $L^{F}  $ is the Lagrangian of the ordinary matter fields $A^{\beta}$. The details of the variations of $S$ with respect to the variables $g^{\alpha\beta}$, $ u^{\nu}$ and $f $ are given in Appendix A. In particular, variation of $S$ with respect to the inverse metric generates (\ref{MEQ}) and the constraint
\begin{equation}\label{duu}
	\nabla_{\alpha}(u^{\alpha}u^{\beta})=0.
\end{equation}
\par From the definition of $\varPhi_{\alpha\beta}$ in (\ref{Phiab}), the first term of $\pounds_{X}(u_{\alpha}u_{\beta})$, $X^{\lambda}\nabla_{\lambda}(u_{\alpha}u_{\beta})$, vanishes in an affine parameterization since the collinearity $X^{\lambda}=fu^{\lambda}$ generates the geodesic terms $u^{\lambda}\nabla_{\lambda}u_{\alpha}$ and $u^{\lambda}\nabla_{\lambda}u_{\beta}$ that both vanish. Moreover, the variation generates a term $-X^{\lambda}\nabla_{\lambda}(u_{\alpha}u_{\beta}) $ as shown in Appendix A, which cancels the leading term in $\pounds_{X}(u_{\alpha}u_{\beta}) $. Thus, the Lagrangian formulation eliminates the geodesic terms in $\varPhi_{\alpha\beta}$ and it must be expressed as
\begin{equation}
	\varPhi_{\alpha\beta}=\frac{1}{2}(\nabla_{\alpha}X_{\beta}+\nabla_{\beta}X_{\alpha})+u^{\lambda}(u_{\alpha}\nabla_{\beta}X_{\lambda}+u_{\beta}\nabla_{\alpha}X_{\lambda}).
\end{equation}
\par It is important to note that the constraint (\ref{duu}) can be obtained from the proof of the ODT before doing any variations of the action functional, and $\Phi$, the trace of $\varPhi_{\alpha\beta}$, has the global property 
\begin{equation}\label{intPhi}
	\int \Phi\sqrt{-g}d^{4}x=\int\nabla_{\alpha}X_{\beta}(g^{\alpha\beta}+2u^{\alpha}u^{\beta})\sqrt{-g}d^{4}x=0.
\end{equation}Although this integral vanishes, it has the variation $\delta\int\Phi\sqrt{-g}d^{4}x=-\int\varPhi_{\alpha\beta}\delta g^{\alpha\beta}\sqrt{-g}d^{4}x $ with respect to the inverse metric that cannot be obtained by dismissing the first term $\nabla_{\alpha}X^{\alpha}$ as a total divergence, and then performing the variation on the second term.
\par   The action $
S^{EHG}=\frac{c^{3}}{16\pi G}\int (R-\Phi)\sqrt{-g}d^{4}x $ that follows from (\ref{intPhi}) generates the modified Einstein equation with no cosmological constant. If $ \Phi$ is set to the constant 2$\Lambda_{C} $, the Einstein equation with a cosmological constant is obtained accordingly, which contradicts (\ref{intPhi}). Thus, $ \Phi $ dynamically replaces the cosmological constant and the \emph{complete} Einstein equation is \begin{equation}\label{ME}
	G_{\alpha\beta}=\frac{8\pi G}{c^{4}}T_{\alpha\beta}
\end{equation} with
\begin{equation}\label{T}
	T_{\alpha\beta}=\tilde{T}_{\alpha\beta}-\frac{c^{4}}{8\pi G}\varPhi_{\alpha\beta}. 
\end{equation}
The local conservation law, $\nabla^{\alpha}T_{\alpha\beta}$ = 0, follows from the diffeomorphic invariance of MGR and is consistent with  Einstein's fundamenatal postulate. In the vacuum, it follows from (\ref{T}) that $ \nabla^{\alpha}\varPhi_{\alpha\beta}=0$ and the contracted Bianchi identity is satisfied from $G_{\alpha\beta}+\varPhi_{\alpha\beta}=0 $ as required. \par Variation of the action functional $S$ with respect to $u^{\nu}$ generates
\begin{equation}\label{u}
	u_{\nu}=\frac{\partial_{\nu}f}{\Phi}
\end{equation}and variation with respect to $f$, the magnitude of $X^{\alpha}$, yields the constraint
\begin{equation}\label{nu}
	\nabla_{\alpha}u^{\alpha}=0.
\end{equation}
The introduction of $\varPhi_{\alpha\beta}$ into GR yields all results of GR plus additional features attributed to dark matter and dark energy as presented herein. The local gravitational energy-momentum tensor $\varPhi_{\alpha\beta}$ describes the dark Universe geometrically. Its absence in GR is why it has been difficult to understand what dark energy is and why a dark matter profile must be introduced into GR to account for the invisible matter in galaxies known as dark matter. 

\subsection{Important properties of $ \varPhi_{\alpha\beta} $}
\begin{theorem}
	$\varPhi_{\alpha\beta}  $ vanishes if and only if $ X^{\mu} $ is a Killing vector
\end{theorem}
\begin{proof}
	  if $ X^{\mu} $ is a Killing vector, $\nabla_{\alpha}X_{\beta}+\nabla_{\alpha}X_{\alpha}=0$ from which $\varPhi_{\alpha\beta}=-\frac{u_{\alpha }X^{\lambda}\nabla_{\lambda}X_{\beta}}{f}-\frac{u_{\beta }X^{\lambda}\nabla_{\lambda}X_{\alpha}}{f}=0$ in an affine parameterization. Conversely, if the tensor $\varPhi_{\alpha\beta}=0$, it vanishes in all coordinate systems including that defined by $u^{\lambda}=(1,0,0,0)$ and $u_{\lambda}=(-1,0,0,0)$. Then $0=\varPhi_{\alpha\beta}=\frac{1}{2}(\nabla_{\alpha}X_{\beta}+\nabla_{\beta}X_{\alpha})-(\nabla_{\beta}X_{\alpha}+\nabla_{\alpha
	}X_{\beta})$ and $X^{\beta}$ is a Killing vector. However, in general, there are no Killing vector fields unless a particular symmetry is involved.
\end{proof}
\par Using (\ref{duu}), it is straightforward to calculate the global property of the trace of the gravitational energy-momentum tensor:
\begin{equation}\label{0phi}
\int g_{\alpha\beta}\varPhi^{\alpha\beta}\sqrt{-g}d^{4}x=\int \Phi \sqrt{-g}d^{4}x=0 
\end{equation} where $ \Phi=\nabla_{\alpha}X_{\beta}(g^{\alpha\beta}+2u^{\alpha}u^{\beta}) $. Equation (\ref{0phi}) means the scalar $ \Phi $ has local positive and negative values, all of which add to zero when integrated over the entire spacetime. $ \Phi $ is globally conserved.  
\section{The conserved energy-momentum tensor $ T_{\alpha\beta} $} 
The invariance of the action functional $S=S^{F}+S^{EH}+S^{G}$ under the symmetry of diffeomorphisms demands a symmetric divergenceless energy-momentum tensor
\begin{equation}\label{Tab}
T^{\alpha\beta}=\tilde{T}^{\alpha\beta}-\frac{c^{4}}{8\pi G}\varPhi^{\alpha\beta}.
\end{equation}  This follows from an analysis of each term in the action functional $ S $. The action $ S^{EH} $ is independently invariant under a diffeomorphism. Variation of the action $ S^{F} $ with respect to the inverse metric 
\begin{equation}
	\delta S^{F}=\int (\frac{\delta L^{F}}{\delta g^{\alpha\beta}}-\frac{1}{2}L^{F}g_{\alpha\beta})\delta g^{\alpha\beta}\sqrt{-g}d^{4}x
\end{equation} generates the symmetric energy-momentum tensor $\tilde{T}_{\alpha\beta}$
\begin{equation}
	\tilde{T}_{\alpha\beta}:=-2c (\frac{\delta L^{F}}{\delta g^{\alpha\beta}}-\frac{1}{2}L^{F}g_{\alpha\beta})  
\end{equation}for gravity minimally coupled  to ordinary matter. The variations of $ S^{F} $ with respect to each field and its derivatives vanish with the corresponding Euler-Lagrange equations.  Variation of $ S^{G} $ with respect to the metric yields $ a\int\varPhi_{\alpha\beta}\delta g^{\alpha\beta}\sqrt{-g}d^{4}x $. Therefore, we can write
\begin{equation}\label{key}
\int (-\frac{1}{2c}\tilde{T}^{\alpha\beta}+a\varPhi^{\alpha\beta})\delta g_{\alpha\beta}\sqrt{-g}d^{4}x=0
\end{equation} where $a=\frac{c^{3}}{16\pi G}$. Under a diffeomorphism, the Lie derivative of the metric along a regular vector $ Y^{\beta} $ generates the infinitesimal change in the metric $\delta g_{\alpha\beta}=\nabla_{\alpha}Y_{\beta}+\nabla_{\beta}Y_{\alpha}$. Integrating by parts then gives
\begin{equation}\label{key}
\int \nabla_{\alpha}(-\frac{1}{2c}\tilde{T}^{\alpha\beta}+a\varPhi^{\alpha\beta})Y_{\beta}\sqrt{-g}d^{4}x=0
\end{equation} which requires
\begin{equation}\label{Tcons}
\nabla_{\alpha} T^{\alpha\beta}=0  
\end{equation}
for diffeomorphisms generated by $ Y^{\beta} $. \par  Equation (\ref{Tcons}) is the local description of the conservation of energy and momentum in a modified theory of GR described by (\ref{MEQ}). The gravitational field has an intrinsic energy-momentum which is attributed to $ \varPhi_{\alpha\beta} $. Being independent of $ \tilde{T}_{\alpha\beta} $, $\frac{c^{4}}{8\pi G}\varPhi_{\alpha\beta} $ provides the additional local energy-momentum of the gravitational field necessary to complete the source $ T_{\alpha\beta} $ of the geometry of spacetime.
\section{Energy-momentum of the gravitational field in the FLRW metric: Dark energy}
A maximally symmetric spacetime is described by the FLRW metric
\begin{equation}
	ds^{2}=-c^{2}dt^{2}+a(t)^{2}[\frac{1}{1-\kappa r^{2}}dr^{2}+r^{2}(d\theta^{2}+{sin^{2}\theta} d\varphi^{2})]
\end{equation} where $a(t)$ is the cosmological scale factor, which satisfies $a>0 $ after the Big Bang at $ t=0 $.  The parameter $\kappa$ is used to describe a particular spatial geometry; $\kappa=1,0,-1$ describes a closed, flat, or open space with constant curvature, respectively. A maximally symmetric second rank (0,2) tensor $B_{\alpha\beta}$ has the components 
\begin{equation}
	B_{00}=\alpha(t),\enspace B_{0j}=0,\enspace B_{ij}=\beta(t)g_{ij}\;\;i,j=1,2,3
\end{equation} where $\alpha(t)$ and $\beta(t)$ are arbitrary functions of time.  We therefore set $\tilde{T}_{00}=c^{2}\varrho,\enspace \tilde{T}_{ij}=pg_{ij},\enspace \tilde{T}^{\mu}_{\mu}=-c^{2}\varrho+3p$ where $\varrho(t)$ and $p(t)$ are designated as the mass density and pressure functions of \emph{ordinary} matter, respectively.
\par  $\varPhi_{ij}$ cannot be expressed in terms of the maximally symmetric form $ \beta g_{ij}$ because the constraints (\ref{u}) and (\ref{nu}) would force $\varPhi=0$. This follows from contracting (\ref{phiu}), which is obtained from  (\ref{u}) and (\ref{nu}), with $u^{\alpha}$ to obtain $u^{\alpha}\varPhi_{\alpha\beta}=\Phi u_{\beta}$. If $\varPhi_{ij}=\beta g_{ij}$, then $\beta u_{i}=\Phi u_{i}$ and $\beta=\Phi$ since $u_{i}\neq0$. Similarly, $u^{0}\varPhi_{00}=-u_{0}\varPhi_{00}=u_{0}\Phi$ so $\Phi=-\varPhi_{00}$. Since $\varPhi_{1}^{1}=\varPhi_{11}g^{11}=\beta$ and $\Phi=-\varPhi_{00}+3\varPhi_{1}^{1}$ from (\ref{Phii}), and $\beta=0$. Thus, $\Phi=0$.
\par To obtain the Friedmann equations, we use the trace of the complete Einstein equation $-\frac{8\pi G}{c^{4}}\tilde{T}-R+\Phi=0$ to rewrite it as
\begin{equation}\label{Rab}
	R_{\alpha\beta}=\frac{8\pi G}{c^{4}} (\tilde{T}_{\alpha\beta}-\frac{1}{2}g_{\alpha\beta}\tilde{T})+\frac{1}{2}g_{\alpha\beta}\Phi -\varPhi_{\alpha\beta}
\end{equation} from which we obtain
\begin{equation}\label{F1}
	\frac{\ddot{a}}{a}=-\frac{4\pi G}{3}(\varrho+\frac{3p}{c^{2}})+\frac{c^{2}}{6}(2\varPhi_{00}+\Phi)
\end{equation} from the $R_{00}$ component. The $R_{11}$ component gives
\begin{equation}\label{frw2}
	\frac{\ddot{a}}{a}+\frac{2\dot{a}^{2}}{a^{2}}+\frac{2\kappa c^{2}}{a^{2}}=4\pi G(\varrho-\frac{p}{c^{2}})+\frac{c^{2}}{2}(\Phi-2\varPhi_{1}^{1})
\end{equation} and the local conservation law $\nabla_{\alpha}T^{\alpha\beta}=0$  yields the conservation equation 
\begin{equation}\label{CE}
	\dot{\varrho}-\frac{c^{2}}{8\pi G}\dot{\varPhi_{00}}=-\frac{3\dot{a}}{a}(\varrho+\frac{p}{c^{2}})+\frac{c^{2}\dot{a}}{8\pi Ga}(\Phi+4\varPhi_{00})
\end{equation}where the dot refers to the derivative with respect to $ t $.
Inserting (\ref{F1}) into (\ref{frw2}) produces a simpler equation
\begin{equation}\label{F2}
	(\frac{\dot{a}}{a})^{2}=\frac{8\pi G\varrho}{3}-\kappa\frac{c^{2}}{a^{2}}+\frac{c^{2}}{6}(\Phi-\varPhi_{00})-\frac{c^{2}\varPhi_{1}^{1}}{2}. 
\end{equation} 
Since  $R_{ij}=(\frac{\ddot{a}}{ac^{2}}+\frac{2\dot{a}^{2}}{a^{2}c^{2}}+\frac{2\kappa }{a^{2}})g_{ij}$, $R_{1}^{1}=R_{2}^{2}=R_{3}^{3}$ and it follows from (\ref{Rab}) that
\begin{equation}\label{Phii}
	\Phi_{1}^{1}=\Phi_{2}^{2}=\Phi_{3}^{3},
\end{equation}
which allows (\ref{F2}) to be written as
\begin{equation}\label{F2s}
	(\frac{\dot{a}}{a})^{2}=\frac{8\pi G\varrho}{3}-\kappa\frac{c^{2}}{a^{2}}-\frac{c^{2}}{3}\varPhi_{00}.
\end{equation} 
Equations (\ref{F1}) and (\ref{F2s}) are the Friedmann equations in MGR. 
\par Dark energy is believed to be responsible for the observed \cite{RiessCC, Perl} accelerated expansion of the current epoch of the Universe. From (\ref{F1}), if $2\varPhi_{00}+\Phi$ is positive and satisfies $2\varPhi_{00}+\Phi>\frac{8\pi G}{c^{2}}(\varrho+\frac{3p}{c^{2}})$, a positive acceleration of the Universe is explained. Since $2\varPhi_{00}+\Phi=\varPhi_{00}+\varPhi_{i}^{i} $, it represents the natural splitting of $\varPhi_{\alpha\beta}$ into its energy and stress components. This suggests dark energy can be described by a positive $\varPhi_{00}$ and dark matter with the stress components $\varPhi_{ij} $ from $\varPhi_{i}^{i}=\varPhi_{ij}g^{ij} $. Then, if the dark energy is large enough relative to the dark matter and the ordinary matter and pressure in a particular epoch of the Universe, it will cause that epoch to accelerate.

\par Thus, the existence of local gravitational energy-momentum in MGR provides the additional freedom necessary to define dark energy and dark matter in terms of $\varPhi_{\alpha\beta}$. Dark energy is the local gravitational energy $\varPhi_{00}$. Since all (0,2) symmetric tensors can be diagonalized, dark matter is described in general by the stress components (pressure) $\varPhi_{ii}$. Dark energy and dark matter are time-dependent in the FLRW metric, and are defined as noted in other metrics where $\varPhi_{00}$ and $\Phi$ are time and space dependent. 
\subsection*{The cyclic Universe}
Cyclic universes have been reported in the literature \cite{Stein,Tur,Baum,Pen}. The unique global property of $\int\Phi\sqrt{-g}d^{4}x=0$ given by (\ref{intPhi}) requires $\Phi$ to have positive and negative values; the Universe must be cyclic. The extrema in $a$ occur when $\dot{a}=0$, which requires the second time derivative of $a$ to satisfy $	\ddot{a}=a(-\frac{4\pi Gp}{c^{2}}-\frac{ c^{2}}{2a^{2}}+\frac{c^{2}\varPhi_{1}^{1}}{2}) $.
The extrema are governed by  $\varPhi_{1}^{1}\lessgtr\frac{8\pi p}{c^{4}}+\frac{1}{a^{2}}$ for a maximum or minimum of $a$, respectively. Dark matter describes the extrema in a cyclic Universe.\par After inflation, the expansion and acceleration continued at a more moderate rate until enough ordinary matter was born. When $\varPhi_{1}^{1}<\frac{8\pi p}{c^{4}}+\frac{1}{a^{2}}$, the cosmological scale factor reached a maximum and a period of deceleration began with $\Phi<\frac{8\pi G}{c^{2}}(\varrho+\frac{3p}{c^{2}})-2\varPhi_{00}$. After some time, a minimum in the decelerating epoch was reached when $\varPhi_{1}^{1}>\frac{8\pi p}{c^{4}}+\frac{1}{a^{2}}$ and the Universe underwent another expansion and acceleration with $\Phi>\frac{8\pi G}{c^{2}}(\varrho+\frac{3p}{c^{2}})-2\varPhi_{00}$. It is known from observations \cite{Rie} that a deceleration preceded the present acceleration of the Universe. $\Phi$ keeps the Universe gravitationally in balance. It cycles to eternity; there can be no Big Rip where it gets torn apart and no Big Crunch where it is destroyed by contraction.
\section{Energy-momentum of the gravitational field: Dark matter}
The $\Lambda$CDM model describes the formation of galaxies after the Big Bang from cooled baryonic matter gravitationally attracted into a dark matter skeleton. Dark matter in the $\Lambda$CDM model also provides the additional mass required to describe the flat rotation curves observed in many galaxies. However, no dark matter particles have been detected and there have been several attempts to explain the flat rotational curves without dark matter. 
\par  The leading candidate is a \emph{phenomenological} model of Modified Newtonian dynamics (MOND) introduced by Milgrom \cite{Mil}. The Newtonian force F is modified according to
\begin{equation}\label{key}
F=m\mu(\frac{a}{A_{0}})a
\end{equation} where $ A_{0} $ is a fundamental acceleration $ \approx1.2\times10^{-10}m/s^{2} $. $ \mu $ is a function of the ratio of the acceleration relative to $ A_{0} $ which tends
to one for $ a\gg A_{0} $ and tends to $ \frac{a}{A_{0}} $ for $ a\ll A_{0} $. MOND successfully explains many, but not all, mass discrepancies observed in galactic data. However, it has no covariant roots in Einstein's equation or cosmological theory. MOND and $\Lambda$CDM were thoroughly discussed by McGaugh in \cite{McG}.\par  Other alternatives to dark matter were reviewed by Mannheim in \cite{Man} with references therein. In particular, Moffat \cite{Mof} used a nonsymmetric gravitational theory without dark matter to obtain the flat rotation curves of some galaxies. The bimetric theory of Milgrom \cite{Mil2} involved two metrics as independent degrees of freedom to obtain a relativistic formulation of MOND.\par  Different approaches to the missing matter problem include dipolar dark matter, which was introduced by Bernard, Blanchet and Heisenberg in \cite{Bern} to solve the problems of cold dark matter at galactic scales and reproduce the phenomenology of MOND. The theory involves two different species of dark matter particles which are separately coupled to the two metrics of bigravity and are linked together by an internal vector field. In \cite{Ver}, a theory of emergent gravity (EG) which claims a possible breakdown in general relativity, was introduced by Verlinde that provided an explanation for Milgrom's phenomenological fitting formula in reproducing the flattening of rotation curves. Campigotto,  Diaferio and Fatibenec \cite{Cam} showed conformal gravity cannot describe galactic rotation curves without the aid of dark matter. On the other hand, a logical analysis based on observational data was presented by Kroupa in \cite{Kro} to support the conjecture that dark matter does not exist.\par The existence of dark matter is based on the assumption that general relativity is correct. However, Einstein's equation is incomplete without the tensor $ \varPhi_{\alpha\beta} $ describing the local energy-momentum of the gravitational field. The validity of modified general relativity is now tested through an attempt to  geometrically describe the additional gravitational attraction in galaxies attributed to dark matter.
\subsection{Modified General Relativity in a spherically symmetric spacetime}
In a region of spacetime where there is no ordinary matter, $ \tilde{T}_{\alpha\beta}=0 $ and the field equations must satisfy
\begin{equation}\label{EF}
G_{\alpha\beta}+\varPhi_{\alpha\beta}=0.
\end{equation} 
 Spherically symmetric solutions to these nonlinear equations are now investigated with a metric of the form
\begin{equation}\label{g}
	ds^{2}=-e^\nu c^{2}dt^{2}+e^{\lambda}dr^{2}+r^{2}(d\theta^{2}+sin^{2}\theta d\varphi^{2})
\end{equation} where $ \nu $ and $ \lambda $ are functions of $r$ and $t$. 
Since the metric is spherically symmetric, there exist three Killing vectors associated with the spherical symmetry of a three-dimensional spatial rotation. $ \varPhi_{\alpha\beta} $ vanishes when $ X^{\beta} $ is a Killing vector and (\ref{EF}) reduces to $ R_{\alpha\beta}=0 $. It follows that $ \nu=-\lambda $ holds in MGR as a direct result of the spherical symmetry invoked to solve (\ref{EF}).
\par  Static solutions to (\ref{EF}) are now sought, which requires $ \partial_{0}X_{\alpha}=0 $ and from the metric $ \partial_{0}\lambda=0,\enspace \partial_{0}\nu=0$. It follows that $ u_{\alpha}\partial_{0}f+f\partial_{0}u_{\alpha}=0 $ and
\begin{equation}\label{uu1}
	u_{\alpha}u^{\alpha}=-1
\end{equation}
 holds. 
\par With $X_{3}=0$ and $\mu:=(1+2u_{0}u^{0})$, the components of $ \varPhi_{\alpha\beta} $ to be considered are then: 
\begin{equation}\label{Phi00}
	\varPhi_{00}=\frac{\mu}{2}e^{-2\lambda}\lambda^{\prime} X_{1}
\end{equation}
\begin{equation}\label{key}
	\varPhi_{11}=(1+2u_{1}u^{1} )({X_{1}}^{\prime}-\frac{1}{2}\lambda^{\prime}X_{1}),
\end{equation} 
\begin{equation}\label{key}
	\varPhi_{22}=(1+2u_{2}u^{2} )(\partial_{2}X_{2}+re^{-\lambda} X_{1}),  
\end{equation} 
\begin{equation}\label{key}
	\varPhi_{33}=r \sin^{2}\theta e^{-\lambda}X_{1}+\sin\theta \cos\theta X_{2},
\end{equation} the Ricci scalar, which from (\ref{EF}) equals $ \Phi $, is
\begin{equation}\label{R}
	\begin{split}
		R=e^{-\lambda}(\lambda^{\prime\prime}-{\lambda^{\prime}}^{2}+\frac{4}{r}\lambda^{\prime}-\frac{2}{r^{2}})+\frac{2}{r^{2}},
	\end{split}
\end{equation}
and the corresponding components of the Einstein tensor are:
\begin{equation}\label{key}
	\begin{split}
		G_{00}=\frac{1}{r^{2}}e^{-2\lambda}(r\lambda^{\prime}-1)+\frac{e^{-\lambda}}{r^{2}},
	\end{split}
\end{equation}
\begin{equation}\label{key}
	\begin{split}
		G_{11}=\frac{1}{r^{2}}(1-r\lambda^{\prime})-\frac{e^{\lambda}}{r^{2}},
	\end{split}
\end{equation}
\begin{equation}\label{key}
	\begin{split}
		G_{22}=\frac{r^{2}e^{-\lambda}}{2}(-\lambda^{\prime\prime}+{\lambda^{\prime}}^{2}-\frac{2\lambda^{\prime}}{r}),
	\end{split}
\end{equation}
\begin{equation}\label{key}
	\begin{split}
		G_{33}={\sin\theta}^2[\frac{r^{2}e^{-\lambda}}{2}(-\lambda^{\prime \prime}+{\lambda^{\prime}}^{2}-\frac{2\lambda^{\prime}}{r})]
	\end{split}
\end{equation} where the prime denotes $ \partial_{1} $.  \par Since $ e^{2\lambda}(\varPhi_{00}+G_{00})+\varPhi_{11}+G_{11}=0 $ from (\ref{EF}),
\begin{equation}\label{phiG0011}
	\begin{split}
		\mu\frac{\lambda^{\prime}}{2}X_{1}+(X_{1}^{\prime}-\frac{\lambda^{\prime}}{2}X_{1})(1+2u_{1}u^{1})=0.
	\end{split}
\end{equation} 
\par  $G_{22}+\varPhi_{22}=0 $ gives \begin{equation}\label{phiG22}
	\begin{split}
		-\lambda^{\prime\prime}+\lambda^{\prime2}-\frac{2}{r}\lambda^{\prime}+\frac{2e^{\lambda}}{r^{2}}(\partial_{2}X_{2}+re^{-\lambda}X_{1})(1+2u_{2}u^{2})=0
	\end{split}
\end{equation}

and $G_{33}+\varPhi_{33}=0 $ in the interval $ 0<\theta<\pi $ yields
\begin{equation}\label{phiG33}
	\begin{split}
		-\lambda^{\prime\prime}+\lambda^{\prime 2}-\frac{2}{r}\lambda^{\prime}+\frac{2e^{\lambda}}{r^{2}}(re^{-\lambda}X_{1}+X_{2}\cot\theta)=0.
	\end{split}
\end{equation} Subtracting (\ref{phiG22}) from (\ref{phiG33}) requires
\begin{equation}\label{C2}
	X_{2}\cot\theta -\partial_{2}X_{2}-2u_{2}u^{2}(\partial_{2}X_{2}+re^{-\lambda}X_{1})=0.
\end{equation} 
\par Using $u_{\alpha}u^{\alpha}=-1$, an equation involving the two independent components $ X_{1} $ and $ X_{2} $ can be obtained from (\ref{phiG0011}) and (\ref{C2}):
\begin{equation}\label{C}
	\begin{split}
		X_2\cot\theta-\partial_{2}X_{2}+\mu(1+\frac{\lambda^{\prime}X_{1}}{\lambda^{\prime}X_{1}-2X_{1}^{\prime}})(\partial_{2}X_{2}+re^{-\lambda}X_{1})=0.
	\end{split}
\end{equation}This is the equation from which power series expressions for $ X_{1} $ and $ X_{2} $ are sought, which can then be used in (\ref{phiG33}) to generate $ \lambda $. \par An expression for $ X_{1}$ of the form $ X_{1}=e^{\lambda}P $ is pursued where $ P $ is a polynomial in $ r $. The power series for $X_{1} $ is then assumed to be 
\begin{equation}\label{X11}
	X_{1}=e^{\lambda}(a_{0}+\frac{a_{1}}{r}+\frac{a_{2}}{r^{2}})
\end{equation} where $ a_{0} $, $ a_{1} $ and $ a_{2} $ are real arbitrary parameters. The power series for $ P $ terminates with the $\frac{a_{2}}{r^{2}}$ term, which ensures the gravitational energy density has the Newtonian $ \frac{1}{r^{4}} $ behaviour.
\par  Equation (\ref{C}), with $ X_{1} $ given by (\ref{X11}), becomes 
\begin{equation}\label{N}
	\partial_{2}X_{2}-m(r)X_{2}\cot\theta=n(r)
\end{equation}where $ n(r):=\frac{2\mu rPP^{\prime}}{\lambda^{\prime}P+2P^{\prime}(1-\mu)}$ and $m(r)=\frac{\lambda^{\prime}P+2P^{\prime}}{\lambda^{\prime}P+2P^{\prime}(1-\mu)}$, which has the solution 
\begin{equation}\label{key}
	\begin{split}
		X_{2}=\frac{n}{1-m}\sin\theta\;_{2}F_{1}(\frac{1}{2},\frac{1-m}{2};\frac{3-m}{2};\sin^{2}\theta)+c_{10}\sin^{m}(\theta).
	\end{split}
\end{equation}$_{2}F_{1}(\alpha,\beta;\gamma;z)  $ is the Gaussian hypergeometric function with $ \alpha=\frac{1}{2} $, $ \beta=-\frac{1-m}{2} $, $ \gamma=\frac{3-m}{2} $ and $ z=\sin^{2}\theta $. Setting the arbitrary constant $ c_{10} $ to zero and using the Euler transformation $_{2}F_{1}(\alpha,\beta;\gamma;z)=(1-z)^{-\alpha}\;  _{2}F_{1}(\alpha,\gamma-\beta;\gamma;\frac{z}{z-1})   $ for $0<\theta<\frac{\pi}{4}  $ yields
\begin{equation}\label{Fabc}
	\begin{split}
		X_{2}&=\frac{n}{1-m}\tan\theta\;_{2}F_{1}(\frac{1}{2},1;\frac{3-m}{2};-\tan^{2}\theta)\\
		&=-rP\tan\theta(1+\Sigma_{n=1}^{\infty}\frac{(\alpha)_{n}(\beta)_{n}(-\tan^{2}\theta)^{n}}{(\gamma)_{n}n!})\\
		&=-rP\tan\theta(1+\frac{1}{3-m}(-\tan^{2}\theta) +O(\tan^{4}\theta) )
	\end{split}
\end{equation}where $ \alpha=\frac{1}{2} $, $ \beta=1 $, $ \gamma=\frac{3-m}{2} $ and $ (\alpha)_{n}=\alpha(\alpha+1)...(\alpha+n-1)\enspace n>0 $ is the Pochhammer symbol. Using the first term of (\ref{Fabc}) as a solution for $ X_{2} $ in (\ref{phiG33}) gives the Schwarzschild solution. 
Expanding $\frac{1}{3-m}:=q(r)$ as a series of inverse powers of $r$ and taking  $\tan\theta<<1$ gives $X_{2}\simeq\tan\theta(b_{2}r+b_{0}+\frac{b_{1}}{r}+O(r^{-2}))$ where the $b_{j}$ are products or sums of products of $a_{j}$ and $q_{j}$ and are therefore independent of the $a_{j}$ in $P$. The second and higher terms of (\ref{Fabc}) consist of powers of $\tan^{2}\theta$ and can be neglected for small $ \tan\theta $. The radial dependence of the expansion of $ X_{2} $ is then then assumed to be of the form
$ b_{2}r+b_{0}+\frac{b_{1}}{r} $ where $b_{2}$, $b_{0} $ and $ b_1 $ are real arbitrary parameters. 
\par Assuming $ b_{2}\neq0$ and using (\ref{uu1})  with $X_{\alpha}=fu_{\alpha} $ where $ f\neq0 $ is the magnitude of $ X_{\alpha}$, it follows that $\frac{b_{2}^{2}\tan^{2}\theta}{f^{2}}\approx-1   $ for extremely large but finite $ r $, which is impossible. Thus, $ b_{2} $ must vanish and $ X_{2} $ does not blow-up as $ r $ becomes extremely large. The physically relevant approximate solution for $ X_{2} $ is then
\begin{equation}\label{X2}
	X_{2}=(b_{0}+\frac{b_{1}}{r})\tan\theta,\;0<\tan\theta<<1.
\end{equation} \par Equation $(\ref{phiG33})$ can now be written as 
\begin{equation}\label{phiG33a}
	\begin{split}
		-\lambda^{\prime\prime}+\lambda^{\prime 2}-\frac{2}{r}\lambda^{\prime}+\frac{2e^{\lambda}}{r^{2}}(a_{0}r+a_{1}+\frac{a_{2}}{r}+b_{0}+\frac{b_{1}}{r})=0.
	\end{split}
\end{equation}
By demanding (in hindsight) 
\begin{equation}\label{b1}
	b_{1}=-a_{2} 
\end{equation}the undesirable term with $ \frac{\ln r}{r} $ in the solution for $ \lambda $ can be avoided. Equation (\ref{phiG33a}) then simplifies to
\begin{equation}\label{main}
	\begin{split}
		-\lambda^{\prime\prime}+\lambda^{\prime 2}-\frac{2}{r}\lambda^{\prime}+\frac{2e^{\lambda}}{r^{2}}(a_{0}r-b)=0
	\end{split}
\end{equation}where
\begin{equation}\label{b}
	-b=a_{1}+b_{0},
\end{equation} which has the exact solution
\begin{equation}\label{wow}
	\begin{split}
		\lambda=-\ln(-a_{0}r+2b\ln r+\frac{c_{1}}{r}+c_{2}), \enspace 0<r<\infty
	\end{split}
\end{equation}where $c_{1}$ and $ c_{2}  $ are arbitrary parameters. Equation (\ref{wow}) represents the extended Schwarzschild solution. It is important to note that since the line element vectors in the line subbundle are bounded and non-vanishing, $ r>0 $. Moreover, the extended Schwarzschild solution demands $ r $ to be finite; $ r $ can be as large as necessary to describe any physically reasonable Universe or part thereof, but it cannot extend to infinity.
\subsubsection{The energy density of the gravitational field} The total energy-momentum tensor given by (\ref{Tab}) requires the energy density of the gravitational field to be described by
\begin{equation}\label{W}
W=-\frac{c^{4}}{8\pi G}\varPhi_{00}.
\end{equation} where
\begin{equation}
\varPhi_{00}=\frac{\mu}{2}X_{1}e^{-2\lambda}\lambda^{\prime}.
\end{equation}From (\ref{X11}) and (\ref{wow}), 
\begin{equation}\label{key}
\varPhi_{00}=\frac{\mu}{2}(a_{0}-\frac{2b}{r}+\frac{c_{1}}{r^{2}})(a_{0}+\frac{a_{1}}{r}+\frac{a_{2}}{r^{2}}),
\end{equation}which must yield the $ \frac{1}{r^{4}} $ behavior of the Newtonian gravitational energy density. By setting each of the coefficients of the $ \frac{1}{r}, \frac{1}{r^{2}},\frac{1}{r^{3}} $ terms to zero, we obtain $ 2b=a_{1} $, $ a_{0}(a_{2}+c_{1})=2ba_{1} $, $ 2ba_{2}=c_{1}a_{1} $ respectively. Then it follows that \begin{equation}\label{a2}
a_{2}=c_{1},
\end{equation}
\begin{equation}\label{a1}
a_{1}^{2}=2\mid a_{0}\mid \mid c_{1}\mid,
\end{equation}and
\begin{equation}\label{b}
b=\pm\sqrt{\frac{\mid a_{0}\mid \mid c_{1}\mid}{2}}.
\end{equation}Hence,
\begin{equation}\label{gen}
\varPhi_{00}=\mu(\frac{1}{2}a_{0}^{2}+\frac{c_{1}^{2}}{2r^{4}})
\end{equation}where $ c_{1} $ is chosen to be the parameter
\begin{equation}\label{c1}
c_{1}=-\frac{2GM}{c^{2}}
\end{equation}from the Schwarzschild solution. From (\ref{W}), it follows that the energy density of the static gravitational field is
\begin{equation}\label{W0}
W=-\mu(\frac{GM^{2}}{4\pi r^{4}}-\frac{c^{4}a_0^{2}}{16\pi G }).
\end{equation}If this static spherically symmetric system satisfies $u_{0}u^{0}=u_{i}u^{i}\;\;i=1,2,3 $\:\:no sum on i, then $u_{0}u^{0}=-\frac{1}{4}$ and $\mu=\frac{1}{2}$. The radial term equals the Newtonian gravitational energy density and that calculated in GR from the weak field approximation: $-\frac{GM^{2}}{8\pi r^{4}}$. In a comoving frame, $u^{0}=1$, $u_{0}=-1$, $u^{i}=u_{i}=0$ and $\mu=-1$. Gravity gravitates and we see that the radial term is twice the Newtonian gravitational energy density. This is not surprising considering the fact that in general relativity, it is calculated from the linearized field equations \cite{LL}. The perturbation expansion of the field equations in terms of $ h_{\alpha\beta} $, a very small change in the metric relative to the flat spacetime Minkowski metric, contains an infinite number of very small terms involving $ h^{2},h^{3},...h\partial h,...h\partial\partial h... $ to all possible powers of $ h$ and its first two derivatives. These truncated self interactions of the perturbed field and the linear term, apparently contribute the same amounts of gravitational energy to $W$. But there is no way of knowing that in GR because it has no tensor that explicitly represents the energy-momentum of the gravitational field. That $W$ is positive in a comoving frame is necessary to describe the local energy density of a gravitational wave: $W=(\frac{GM^{2}}{4\pi r^{4}}-\frac{c^{4}a_0^{2}}{16\pi G })$
\subsubsection{The radial force and galactic rotation curves}
The radial force on an object of mass m can now be calculated from (\ref{wow}) with $c_{2}=1$. Using the conventional relationship of the Newtonian potential $ \phi $ to $ g_{00} $, 
\begin{equation}\label{key}
\phi=\frac{c^{2}}{2}(e^{\nu}-1),
\end{equation} the radial force $ F_{r} $ is
\begin{equation}\label{Fr}
\begin{split}
F_{r}&=-m\partial_{1}\phi\\
&=\frac{mc^{2}}{2}(\frac{c_{1}}{r^{2}}-\frac{2b}{r}+\mid a_{0}\mid).
\end{split}
\end{equation} From (\ref{c1}) where M represents the total mass of the galaxy, we arrive at the modified Newtonian force 
\begin{equation}\label{Newt}
F_{r}=-\frac{GMm}{r^{2}}-\frac{\mid b\mid mc^{2}}{r}+\frac{\mid a_{0}\mid mc^{2}}{2}.
\end{equation} The correction terms to the Newtonian force come from the non-zero components of the line element field in the energy-momentum tensor $ \varPhi_{\alpha\beta}$. The second term is gravitationally attractive and represents the ``dark matter" correction. The components of the line element field can change their sign, which means the parameter $ b $ has two signs; its absolute value is invoked to ensure that term is negative. It is the term that gives rise to the flat rotation curves. The third term is positive and repulsive. Dark energy is repulsive in the present accelerating epoch. However, a previous decelerating epoch when dark energy was not dominant was observed by Riess et al. \cite{Rie}. They used the Hubble telescope to provide the first conclusive evidence for cosmic deceleration that preceded the current epoch of cosmic acceleration.    \par Assuming a circular orbit about a point mass, it follows that the orbital velocity of a star rotating in the galaxy satisfies
\begin{equation}\label{V}
v^{2}=\frac{GM(r)}{r}+\mid b\mid c^{2}-\frac{\mid a_{0}\mid c^{2}}{2}r
\end{equation} Equation (\ref{V}) demands an upper limit to r describing a large but finite galaxy.  
\par  Because $ a_{0}\neq0$, it is possible for the Newtonian force to balance the dark energy force. This requires
\begin{equation}\label{VF}
\frac{GM(r)}{r}-\frac{\mid a_{0}\mid c^{2}}{2}r=0
\end{equation} and
\begin{equation}\label{v2}
v^{2}=\mid b\mid c^{2}
\end{equation} describes a specific class of galaxies with a flat orbital rotation curve. From (\ref{VF}) and (\ref{v2}), we obtain the Tully-Fisher relation
\begin{equation}\label{TF}
v^{4}=\frac{GMc^{2}\mid a_{0}\mid }{2}.
\end{equation} This result holds for any finite r in contrast to EG which holds only for large r as determined by Lelli, McGaugh and Schombert \cite{Lell}. With $ \frac{{c^{2}\mid a_{0}}\mid}{2}:=A_{0}$, the Tully-Fisher relation in MOND is evident. However, $A_{0}$ is a fixed and $a_{0}$ is not, which shows MOND is limited relative to MGR.
\par   The importance of the radial acceleration relative to the rotation curves of galaxies was discussed by Lelli, McGaugh, Schombert, and Pawlowski in \cite{Lell2} where it was determined that late time galaxies (spirals and irregulars), early time galaxies (ellipticals and lenticulars), and the most luminous dwarf spheroidals follow a tight radial acceleration relation which correlates well with that due to the distribution
of baryons. \par  Equation (\ref{V}) is general enough to describe the rotation curves of many types of galaxies. It can be simply expressed as
\begin{equation}\label{Dr}
	v^{2}=\mid b\mid c^{2}+r\Delta (r) 	
\end{equation} where $ \Delta(r):=\frac{GM(r)}{r^{2}}-\frac{\mid a_{0}\mid c^{2}}{2} $ is a function of $ r $ that represents the difference between the Newtonian and dark energy forces. In a region of the galaxy where these forces are slightly different, (\ref{Dr}) describes the linear rise or fall of the rotation curve. For example, galaxy NGC4261 has a relatively flat rotation curve but starts to rise at larger radii, reaching velocities of 700 km $s^{-1}$ at 100 kpc \cite{Lell2}. That requires $\Delta(r) $ to be positive. As another example, both $ \mid a_{0}\mid c^{2}r $ and $ \mid b\mid c^{2} $ could be small enough relative to $ \frac{GM}{r} $ so that the  Newtonian term is dominant. Galaxies with no flat rotation curves have recently been observed by van Dokkum et.al \cite{Van}.
\par From the expression for $ \lambda $ given by (\ref{wow}), it is clear that in the absence of the line element covectors in GR, the Lovelock tensors $G_{\alpha\beta}$ and $\Lambda g_{\alpha\beta}$ cannot represent the additional gravitational attraction attributed to dark matter. This explicitly shows why GR is incomplete and why dark matter particles were invoked in the $ \Lambda $CDM model to describe dark matter.
\section{Conclusion}
The results in this article stem from the existence of the line element field in a Lorentzian spacetime. It is a fundamental part of all Lorentzian metrics. $\varPhi_{\alpha\beta}$ is constructed from the Lie derivative of a particular Lorentzian metric and its associated unit line element covectors along the flow of the corresponding line element vector. $\varPhi_{\alpha\beta}$, which is absent in GR, solves the problem of the localization of gravitational energy-momentum. It completes the Einstein equation and leaves it intact in form. The energy-momentum of the gravitational field is localizable and measurable.  \par The line element vectors in the line subbundle are bounded and non-vanishing, so the solutions to the complete Einstein equation are bounded. $\varPhi_{\alpha\beta}$ vanishes if and only if there is a Killing vector associated with some symmetry. The cosmological constant $ \Lambda $ is dynamically replaced by $ \Phi $, and $\Phi $ has the global property $ \int g^{\alpha\beta}\varPhi_{\alpha\beta}\sqrt{-g}d^{4}x=0  $, which demands the Universe to be cyclic.
\par The existence of local gravitational energy-momentum provides an explanation for dark energy and dark matter: dark energy is the local gravitational energy $\varPhi_{00}$ and dark matter is described in general by the stress components (pressure) $\varPhi_{ii}$. 
\par  A static solution is obtained from the modified Einstein equation in a spherically symmetric metric. The modified Newtonian force contains two additional terms: one represents the dark energy force which depends on the parameter $ a_{0} $, and the other represents the $ \frac{1}{r} $ ``dark matter" force which depends on the parameter $ b $; dark matter is explained geometrically. The Tully-Fisher relation is obtained by balancing the dark energy force with the Newtonian force. This condition describes the class of galaxies associated with flat rotation curves. The rotation curves for galaxies with no flat orbital curves, and those with rising rotation curves for large radii describe examples of the flexibility of the orbital rotation curve equation. The results obtained from the \emph{complete} Einstein equation thus far are able to substantially describe the missing mass problem attributed to dark matter. Further mathematical and detailed numerical analyses to explore the ability of the energy-momentum tensor of the gravitational field to replace dark matter in cosmology, are fully warranted. This rigorous analysis with comparison to astronomical data may still point to the existence of dark matter particles to some extent. But even if that is the case, the gravitational role of dark matter is substantially reduced by the impact of the energy-momentum tensor of the gravitational field.
\par  Thus, $ \varPhi_{\alpha\beta} $ represents the local energy-momentum of the gravitational field and explains particular features of dark energy and dark matter. It is the symmetric tensor that Einstein sought many years ago. 
\section*{Acknowledgements} 
I would like to thank the anonymous referee for his/her constructive comments. 

\appendix
\numberwithin{equation}{section}
\section {Variations of the action functional}
There are three variables in MGR: $g^{\alpha\beta}$, $ X^{\beta} $ and $ u^{\beta} $. However, since $\bm{X}$ and $\bm{u}$ are collinear, $ X^{\beta}=fu^{\beta} $ and $X_{\beta}=fu_{\beta}  $ where $ f\neq0 $ is the magnitude of  both $ \bm{X} $ and its covector. Strictly speaking, the independent variables are $g^{+\alpha\beta}$, $ u^{\beta} $ and $ f $. However, since $ \frac{\delta }{\delta g^{+\alpha\beta}}=\frac{\delta }{\delta g^{\mu\nu}}\frac{\delta g^{\mu\nu}}{\delta g^{+\alpha\beta}}=\frac{\delta }{\delta g^{\alpha\beta}} $, $ g^{\alpha\beta} $ can be treated as an independent variable in the variation with respect to the inverse metric. Variations of $ S $ with respect to $g^{\alpha\beta}$, $ u^{\beta} $ and $ f $ are developed as follows:
\subsection*{Variation of $S^{G}$ with respect to $g^{\alpha\beta}$}
Variation of $ S^{EH} $ is well known from any textbook on GR. What needs to be established is the variation of $ S^{G} $ with respect to the inverse metric where 
$ S^{G}=-a\int\Phi\sqrt{-g}d^{4} $. The parameter $ a $ is not needed in this calculation, but does belong in $ S $. $ \Phi $ can be expressed as $\Phi=\nabla_{\alpha}X_{\beta}(g^{\alpha\beta}+2u^{\alpha}u^{\beta})=\nabla_{\alpha}X_{\beta}g^{+\alpha\beta}  $ so
\begin{equation}\label{dgPhi}
	\begin{split}
		-\delta \int\Phi\sqrt{-g}d^{4}x
		=-\int\delta (\nabla_{\alpha}X_{\beta})g^{+\alpha\beta}\sqrt{-g}d^{4}x-\int\nabla_{\alpha}X_{\beta}(\delta g^{\alpha\beta}+2\delta (u^{\alpha}u^{\beta}))\sqrt{-g}d^{4}x\\
		-\int\nabla_{\alpha}X_{\beta}(g^{\alpha\beta}+2u^{\alpha}u^{\beta})\delta\sqrt{-g}d^{4}x
	\end{split}
\end{equation}
The third integral in (\ref{dgPhi}) is equivalent to
\begin{equation}\label{3}
	\begin{split}
		\frac{1}{2}\int\nabla_{\mu}X_{\nu}(g^{\mu\nu}+2u^{\mu}u^{\nu})g_{\alpha\beta}\delta g^{\alpha\beta}\sqrt{-g}d^{4}x.
	\end{split}
\end{equation} \par To compute the second integral in (\ref{dgPhi}), consider $g_{\alpha\lambda} u^{\alpha}u^{\beta}=u_{\lambda}u^{\beta} $. In a Lorentzian spacetime $ \bm{u} $ satisfies $ u_{\lambda}u^{\lambda}=-1 $ but in a Riemannian spacetime, $ u_{\lambda}u^{\lambda}=1 $. Then $ \delta(u_{\lambda}u^{\beta})=0 $ 
\footnote{In $ g^{+}_{\alpha\beta}$, $u_{\lambda}u^{\lambda}=1  $ so in $ g_{\alpha\beta} $: $u_{\lambda}u^{\beta}=g^{\beta\mu}u_{\mu}u_{\lambda}\\=(g^{+\beta\mu}-2u^{\beta}u^{\mu})u_{\lambda}u_{\mu}=u_{\lambda}u^{\beta}-2u_{\lambda}u^{\beta}=-u_{\lambda}u^{\beta}$. Thus, if $\lambda=\beta, -1=-1 $ and if $\lambda\neq\beta, u_{\lambda}u^{\beta}=0  $ and $ \delta(u_{\lambda}u^{\beta})=0 $.} and $ \delta(u^{\alpha}u^{\beta})g_{\alpha\lambda}=-u^{\alpha}u^{\beta}\delta g_{\alpha\lambda} $.
Contracting that with $ g^{\lambda\rho} $ and re-labelling the indicies gives
\begin{equation}\label{deltauu}
	\delta(u^{\alpha}u^{\beta})=u^{\beta}u_{\lambda}\delta g^{\alpha\lambda}.
\end{equation}The second integral can then be expressed as $ -\int\nabla_{\alpha}X_{\beta}(\delta g^{\alpha\beta}+2u^{\beta}u_{\lambda}\delta g^{\alpha\lambda})\sqrt{-g}d^{4}x $. Keeping $ \alpha $ and $ \beta $ fixed while summing over $ \lambda $ in the integrand: if $ \lambda=\beta$, $ \nabla_{\alpha}X_{\beta}(\delta g^{\alpha\beta}+2u^{\lambda}u_{\lambda}\delta g^{\alpha\beta})=-\nabla_{\alpha}X_{\beta}\delta g^{\alpha\beta}$, and when $ \lambda\neq\beta $ the integrand is $\nabla_{\alpha}X_{\beta}\delta g^{\alpha\beta}  $ so the integrand vanishes in the sum over all $ \lambda $.
\par The first integral in (\ref{dgPhi}) is now calculated.
\begin{equation}\label{1}
	\begin{split}
		-\int\delta(\nabla_{\alpha}X_{\beta})g^{+\alpha\beta}\sqrt{-g}d^{4}x=-\int(\nabla_{\alpha}\delta X_{\beta}-X_{\lambda}\delta\Gamma^{\lambda}_{\alpha\beta})g^{+\alpha\beta}\sqrt{-g}d^{4}x\\
		=2\int\nabla_{\alpha}(u^{\alpha}u^{\beta})\delta X_{\beta}\sqrt{-g}d^{4}x+\int X_{\lambda}\delta\Gamma^{\lambda}_{\alpha\beta}g^{+\alpha\beta}\sqrt{-g}d^{4}x
	\end{split}
\end{equation}integrating by parts. The second integral in (\ref{1}) is expressed as
\begin{equation}\label{I2}
	\frac{1}{2}\int X_{\lambda}g^{+\alpha\beta}g^{\lambda\rho}(\nabla_{\alpha}\delta g_{\rho\beta}+\nabla_{\beta}\delta g_{\rho\alpha}-\nabla_{\rho}\delta g_{\alpha\beta})\sqrt{-g}d^{4}x.
\end{equation}To compute this integral, we need to use:\newline $ \delta(g^{+\alpha\beta}g_{\alpha\lambda})=\delta(\delta^{\beta}_{\lambda}+2u_{\lambda}u^{\beta})=0 $ so $ \delta g^{+\alpha\beta}g_{\alpha\lambda}=-g^{+\alpha\beta}\delta g_{\alpha\lambda} $. Contracting that with $ g^{\lambda\rho} $ gives
\begin{equation}\label{g+1}
	\delta g^{+\rho\beta}=-g^{+\alpha\beta}g^{\lambda\rho}\delta g_{\alpha\lambda}.
\end{equation}
Continuing with (\ref{I2}), we can use the Lorentz connection $ \nabla $ with $ g^{+} $ in the first two terms because it can operate on the divergence of a symmetric tensor as shown in the ODT:
\begin{equation}\label{key}
	\begin{split}
		-\frac{1}{2}\int X_{\lambda}
		(\nabla_{\alpha}\delta g^{+\alpha\lambda}+\nabla_{\beta}\delta g^{+\lambda\beta})\sqrt{-g}d^{4}x-\frac{1}{2}\int X_{\lambda}(g^{\alpha\beta}+2u^{\alpha}u^{\beta})g^{\lambda\rho}\nabla_{\rho}\delta g_{\alpha\beta}\sqrt{-g}d^{4}x\\
		=-\int X_{\lambda}
		\nabla_{\alpha}\delta g^{+\alpha\lambda}\sqrt{-g}d^{4}x+\frac{1}{2}\int X_{\lambda}g_{\alpha\beta}\nabla^{\lambda}\delta g^{\alpha\beta}\sqrt{-g}d^{4}x-\int X_{\lambda}u^{\alpha}u^{\beta}\nabla^{\lambda}\delta g_{\alpha\beta}\sqrt{-g}d^{4}x\\
		=-\int X_{\beta}
		\nabla_{\alpha}\delta (g^{\alpha\beta}+2u^{\alpha}u^{\beta})\sqrt{-g}d^{4}x+\frac{1}{2}\int X_{\lambda}g_{\alpha\beta}\nabla^{\lambda}\delta g^{\alpha\beta}\sqrt{-g}d^{4}x\\+\int X_{\lambda}u_{\alpha}u_{\beta}\nabla^{\lambda}\delta g^{\alpha\beta}\sqrt{-g}d^{4}x
	\end{split}
\end{equation}using $ u^{\alpha}u^{\beta}\nabla^{\lambda}\delta g_{\alpha\beta}=u_{\mu}u_{\nu}\nabla^{\lambda}g^{\alpha\mu}g^{\beta\nu}\delta g_{\alpha\beta}=-u_{\alpha}u_{\beta}\nabla^{\lambda}\delta g^{\alpha\beta} $. Integrating by parts gives
\begin{equation}\label{key}
	\begin{split}
		\int [\nabla_{\alpha}X_{\beta}-\frac{1}{2}\nabla^{\lambda}X_{\lambda}g_{\alpha\beta}-\nabla^{\lambda}(X_{\lambda}u_{\alpha}u_{\beta})]\delta g^{\alpha\beta}\sqrt{-g}d^{4}x+2\int \nabla_{\alpha}X_{\beta}\delta (u^{\alpha}u^{\beta})\sqrt{-g}d^{4}x\\
		=\int [\nabla_{\alpha}X_{\beta}-\frac{1}{2}\nabla^{\lambda}X_{\lambda}g_{\alpha\beta}-\nabla^{\lambda}(X_{\lambda}u_{\alpha}u_{\beta})+2u^{\lambda}u_{\beta}\nabla_{\alpha}X_{\lambda}]\delta g^{\alpha\beta}\sqrt{-g}d^{4}x
	\end{split}
\end{equation}from (\ref{deltauu}) and re-labelling indicies. \par The complete variation is then
\begin{equation}\label{key}
	\begin{split}
		\int [\nabla_{\alpha}X_{\beta}-\frac{1}{2}\nabla^{\lambda}X_{\lambda}g_{\alpha\beta}-\nabla^{\lambda}(X_{\lambda}u_{\alpha}u_{\beta})+2u^{\lambda}u_{\beta}\nabla_{\alpha}X_{\lambda}\\+\frac{1}{2}\nabla_{\mu}X_{\nu}(g^{\mu\nu}+2u^{\mu}u^{\nu})g_{\alpha\beta}]\delta g^{\alpha\beta}\sqrt{-g}d^{4}x\\
		+2\int\nabla_{\alpha}(u^{\alpha}u^{\beta})\delta X_{\beta}\sqrt{-g}d^{4}x.
	\end{split}
\end{equation}Since the variations in $ \delta g^{\alpha\beta} $ are independent to those of $ \delta X_{\beta} $, the variation in S will yield
\begin{equation}\label{nuabapp}
	\nabla_{\alpha}(u^{\alpha}u^{\beta})=0
\end{equation} and the contribution from the variation of $ S^{G} $ with respect to $ g^{\alpha\beta} $ is
\begin{equation}\label{key}
	\begin{split}
		\int[\nabla_{\alpha}X_{\beta}+2u^{\lambda}u_{\beta}\nabla_{\alpha}X_{\lambda}-\nabla^{\lambda}(X_{\lambda}u_{\alpha}u_{\beta})+\nabla_{\mu}X_{\nu}u^{\mu}u^{\nu}g_{\alpha\beta}]\delta g^{\alpha\beta}\sqrt{-g}d^{4}x\\
		=\int[\nabla_{\alpha}X_{\beta}+2u^{\lambda}u_{\beta}\nabla_{\alpha}X_{\lambda}-X^{\lambda}\nabla_{\lambda}(u_{\alpha}u_{\beta})+\nabla_{\mu}X_{\nu}(u^{\mu}u^{\nu}g_{\alpha\beta}-u_{\alpha}u_{\beta}g^{\mu\nu})]\delta g^{\alpha\beta}\sqrt{-g}d^{4}x\\
		=\int[\frac{1}{2}(\nabla_{\alpha}X_{\beta}+\nabla_{\beta}X_{\alpha})+u^{\lambda}(u_{\alpha}\nabla_{\beta}X_{\lambda}+u_{\beta}\nabla_{\alpha}X_{\lambda})\\+\nabla_{\mu}X_{\nu}(u^{\mu}u^{\nu}g_{\alpha\beta}-u_{\alpha}u_{\beta}g^{\mu\nu})]\delta g^{\alpha\beta}\sqrt{-g}d^{4}x
	\end{split}
\end{equation}using the symmetry in $ \delta g^{\alpha\beta} $ and setting\footnote{Note that $ \pounds_{X}(u_{\alpha}u_{\beta}) $ has the first term $ X^{\lambda}\nabla_{\lambda}(u_{\alpha}u_{\beta}) $ that will cancel the geodesic term $-X^{\lambda}\nabla_{\lambda}(u_{\alpha}u_{\beta})  $ so the definition of $ \varPhi_{\alpha\beta} $ is consistent without invoking geodesics in the ODT and the variation of $ S^{G} $.} $-X^{\lambda}\nabla_{\lambda}(u_{\alpha}u_{\beta})=-fu^{\lambda}\nabla_{\lambda}(u_{\alpha}u_{\beta})=0  $ in an affine parameterization where $ f\neq0 $ is the magnitude of $ X^{\lambda} $. The last term in the variation vanishes, which follows by writing the tensor in brackets   $-u_{\alpha}u_{\beta}g^{\mu\nu}+u^{\mu}u^{\nu}g_{\alpha\beta}$ as its equivalent, $\frac{1}{2}(g^{{+}{\mu\nu}}g_{\alpha\beta}-g^{+}_{\alpha\beta}g^{\mu\nu}) $, and choosing an orthonormal basis $(e_{\alpha})$ at a point $p\in\mathcal{M} $ for $ g^{+} $ with $ e_{0}=u $. Then, $ u^{0}=u_{0}=1 $, $ u^{i}=u_{i}=0\;(i=1,2,3)$, $g^{+}_{\alpha\beta}=\delta_{\alpha\beta}$, $ g^{00}=-g^{{+}00}$ and $g_{00}=-g^{+}_{00}$, with all other components of the metric g equal to those of the metric $ g^{+}$. From the definition of $\varPhi_{\alpha\beta}  $
\begin{equation}\label{}
	\varPhi_{\alpha\beta}:=\frac{1}{2}(\nabla_{\alpha}X_{\beta}+\nabla_{\beta}X_{\alpha})+u^{\lambda}(u_{\alpha}\nabla_{\beta}X_{\lambda}+u_{\beta}\nabla_{\alpha}X_{\lambda})
\end{equation}it follows that 
$ \delta{S}^{G}=-a\delta \int\Phi\sqrt{-g}d^{4}x=a\int\varPhi_{\alpha\beta}\delta g^{\alpha\beta}\sqrt{-g}d^{4}x.$
\subsection*{Variation of $S$ with respect to $f$}
$ f\neq0 $, the magnitude of $ X^{\beta} $, is independent of $ u^{\beta} $ so variation of $ S $ with respect to $ f $ must be performed. With $ X^{\beta} $ and its first derivative not appearing in the matter energy-momentum tensor, variation with respect to $ f $ involves only $ S^{G} $:
\begin{equation}\label{key}
	\begin{split}
		S^{G}=-a\int[\nabla_{\alpha}(fu^{\alpha})+2u^{\alpha}u^{\beta}\nabla_{\alpha}(fu_{\beta})]\sqrt{-g}d^{4}x\\
		=-a\int(f\nabla_{\alpha}u^{\alpha}-u^{\alpha}\nabla_{\alpha}f)\sqrt{-g}d^{4}x
	\end{split}
\end{equation}
using $ X_{\alpha}=fu_{\alpha}$ and $u^{\beta}\nabla_{\alpha}u_{\beta}=0$. Hence, 
\begin{equation}\label{key}
	\begin{split}
		\delta S^{G}&=-a\int(\nabla_{\alpha}u^{\alpha}\delta f-u^{\alpha}\delta\nabla_{\alpha}f)\sqrt{-g}d^{4}x\\
		&=-2a\int\nabla_{\alpha}u^{\alpha}\delta f\sqrt{-g}d^{4}x.
	\end{split}
\end{equation} It follows that 
\begin{equation}\label{bnu0}
	\nabla_{\alpha}u^{\alpha}=0 
\end{equation} for arbitrary variations of f. $ \nabla_{\alpha}u^{\alpha}=0 $ can also be obtained from (\ref{duu}) in an affine parameterization.
\subsection*{Variation of $S$ with respect to $u^{\nu}$} 
The action functionals $S^{F}$ and $S^{EH}$ do not depend on $X^{\mu}$ and  $\frac{\delta X^{\mu}}{\delta u^{\nu}}=f\delta^{\mu}_{\nu}$ is well defined, so the variation with respect to $u^{\nu}$ of $S$ depends only on $S^{G}$:
\begin{equation}\label{Su}
	\frac{\delta S^{G}}{\delta u^{\nu}}=-\int\frac{\delta }{\delta u^{\nu}}(\Phi\sqrt{-g})d^{4}x.
\end{equation}
From the collinearity $X_{\alpha}=fu_{\alpha}$, $\varPhi_{\alpha\beta}$ can be expressed as
\begin{equation}\label{phiu}
	\varPhi_{\alpha\beta}=\frac{f}{2}(\nabla_{\alpha}u_{\beta}+\nabla_{\beta}u_{\alpha})-\frac{1}{2}(u_{\beta}\nabla_{\alpha}f+u_{\alpha}\nabla_{\beta}f).
\end{equation}Using (\ref{bnu0}), $\Phi$ can be expressed as
\begin{equation}\label{Phiu}
	\Phi=-u^{\alpha}\nabla_{\alpha}f.
\end{equation}
Then, (\ref{Su}) becomes
\begin{equation}\label{dSu}
	\frac{\delta S^{G}}{\delta u^{\nu}}=\int[(\delta^{\alpha}_{\nu}\partial_{\alpha}f+u^{\alpha}\frac{\delta}{\delta u^{\nu}}(\partial_{\alpha}f) )\sqrt{-g}-\Phi\frac{\delta\sqrt{-g}}{\delta u^{\nu}}]d^{4}x
\end{equation} The third term is
\begin{equation}\label{sqrtg}
	\begin{split}
		-\Phi\frac{\delta\sqrt{-g}}{\delta u^{\nu}}&=-\Phi\frac{\delta\sqrt{-g}}{\delta g^{\alpha\beta}}\frac{\delta g^{\alpha\beta}}{\delta u^{\nu}}\\
		&=-\Phi g_{\alpha\beta}\sqrt{-g}(\delta^{\alpha}_{\nu}u^{\beta}+\delta^{\beta}_{\nu}u^{\alpha})\\
		&=-2\Phi\sqrt{-g}u_{\nu}
	\end{split}
\end{equation}so (\ref{dSu}) can be written as
\begin{equation}
	\frac{\delta S^{G}}{\delta u^{\nu}}=\int(\partial_{\nu}f+u^{\alpha}\frac{\delta}{\delta u^{\nu}}(\partial_{\alpha}f)-2\Phi u_{\nu})\sqrt{-g})d^{4}x
\end{equation}

Setting $\partial_{\alpha}f=u_{\alpha}P$ for some scalar $P$ requires
\begin{equation}
	u^{\alpha}\frac{\delta}{\delta u^{\nu}}(\partial_{\alpha}f)=-\frac{\delta P}{\delta u^{\nu}}-P u_{\nu}
\end{equation}using 
\begin{equation}
	u^{\alpha}\frac{\delta u_{\alpha}}{\delta u^{\nu}}=-u_{\nu}
\end{equation} from $u_{\alpha}u^{\alpha}=-1$. Thus, 
\begin{equation}
	\frac{\delta S^{G}}{\delta u^{\nu}}=-\int(\frac{\delta \Phi}{\delta u^{\nu}}\sqrt{-g}+\Phi \frac{\delta\sqrt{-g}}{\delta u^{\nu}})d^{4}x=\int(\partial_{\nu}f-\frac{\delta P}{\delta u^{\nu}}-(P+2\Phi) u_{\nu})\sqrt{-g})d^{4}x=0.
\end{equation}Choosing $P=\Phi$ generates
\begin{equation}\label{bunu}
	\partial_{\nu}f=\Phi u_{\nu}.
\end{equation}
\par Note that (\ref{bnu0}) and (\ref{bunu}) are consistent with $ \Phi=\nabla_{\alpha}X^{\alpha}+2u^{\alpha}u^{\beta}\nabla_{\alpha}X_{\beta} $:\begin{equation}\label{PhiPhi}
	\begin{split}
		\Phi&=\nabla_{\alpha}X^{\alpha}+2u^{\alpha}u^{\beta}\nabla_{\alpha}X_{\beta}\\
		&=f\nabla_{\alpha}u^{\alpha}+u^{\alpha}\partial_{\alpha}f+2u^{\alpha}u^{\beta}u_{\beta}\partial_{\alpha}f+2fu^{\alpha}u^{\beta}\nabla_{\alpha}u_{\beta}\\
		&=u^{\alpha}\partial_{\alpha}f-2u^{\alpha}\partial_{\alpha}f\\
		&=-u^{\alpha}u_{\alpha}\Phi\\
		&=\Phi.
	\end{split}
\end{equation}

\end{document}